\documentclass[journal,twoside,web]{ieeecolor}
\usepackage{generic}
\usepackage{amsmath,amssymb,amsfonts}
\usepackage{algorithmic}
\usepackage{graphicx}
\usepackage{algorithm,algorithmic}
\usepackage{hyperref}
\usepackage{textcomp}
\usepackage{generic}
\usepackage{amsmath,amssymb,amsfonts}
\usepackage{algorithmic}
\usepackage{graphicx}
\usepackage{textcomp}
\usepackage{generic}
\usepackage{amsmath,amssymb,amsfonts}
\usepackage{mathrsfs}
\usepackage{algorithmic}
\usepackage{graphicx}
\usepackage{textcomp}
\usepackage{color}
\usepackage{latexsym}
\usepackage{amsfonts,amssymb}
\usepackage{textcomp}
\usepackage{bbm}
\usepackage{dsfont}
\usepackage{fancyhdr}
\usepackage{caption}
\usepackage{algorithm}
\usepackage{cite}

\renewenvironment{proof}{\vspace{.1cm}\noindent{\sc Proof.}\hspace{0.10cm}\,\,}{$\hfill\Box$\vspace{.1cm}} 
\newtheorem{theorem}            {Theorem}[section] 

\newtheorem{definition}         [theorem]{Definition}

\newtheorem{lemma}              [theorem]{Lemma} 
\newtheorem{proposition}		[theorem]{Proposition}

\newtheorem{example}		[theorem]{Example}
\newtheorem{remark}	      [theorem]{Remark}
\newtheorem{fact}         [theorem]{Fact}

\newcommand{\bbm}[1]{\left[\begin{matrix} #1 \end{matrix}\right]}

\def\BibTeX{{\rm B\kern-.05em{\sc i\kern-.025em b}\kern-.08em
    T\kern-.1667em\lower.7ex\hbox{E}\kern-.125emX}}
\begin{document}
\title{Neural network based control of unknown nonlinear systems via contraction analysis}
\author{Hao Yin, Claudio De Persis, Bayu Jayawardhana, Santiago Sánchez Escalonilla Plaza 
\thanks{The authors are with the Jan C. Willems Center for Systems and Control, Faculty of Science and Engineering, University of Groningen, 9747AG Groningen, The Netherlands. {\tt\small hao.yin@rug.nl, c.de.persis@rug.nl, b.jayawardhana@rug.nl, and santiago.sanchez@rug.nl.}}
\thanks{This work was supported by NWO Digital Twin Project 05 - Autonomous Process $\&$ Control Reconfiguration and Optimization (200274).}}

\maketitle

\begin{abstract}
This paper studies the design of neural network (NN)-based controllers for unknown nonlinear systems, using contraction analysis. A Neural Ordinary Differential Equation (NODE) system is constructed by approximating the unknown draft dynamics with a feedforward NN. Incremental sector bounds and contraction theory are applied to the activation functions and the weights of the NN, respectively. It is demonstrated that if the incremental sector bounds and the weights satisfy some non-convex conditions, the NODE system is contractive. To improve computational efficiency, these non-convex conditions are reformulated as convex LMI conditions. Additionally, it is proven that when the NODE system is contractive, the trajectories of the original autonomous system converge to a neighborhood of the unknown equilibrium, with the size of this neighborhood determined by the approximation error. For a single-layer NN, the NODE system is simplified to a continuous-time Hopfield NN. If the NODE system does not satisfy the contraction conditions, an NN-based controller is designed to enforce contractivity. This controller integrates a linear component, which ensures contraction through suitable control gains, and an NN component, which compensates for the NODE system's nonlinearities. This integrated controller guarantees that the trajectories of the original affine system converge to a neighborhood of the unknown equilibrium. The effectiveness of the proposed approach is demonstrated through two illustrative examples.
\end{abstract}

\section{Introduction}
The universal approximation theorem \cite{hornik1989multilayer} states that a multilayer neural network (NN) can approximate any nonlinear function $f(x)$. Due to their powerful approximation capabilities, NNs have garnered significant attention in control theory, particularly in model identification \cite{chen1990non,wang2017new}. However, since control theory primarily focuses on analyzing and designing strategies for dynamical systems to achieve desired behaviors such as convergence, and robustness \cite{khalil2002control}, this raises two fundamental questions: 1. Given a neural network $\mathcal{N}_{\theta}$, can we determine the properties of the Neural Ordinary Differential Equation (NODE) system $\dot{x}=\mathcal{N}_{\theta}$? Moreover, if the properties of the NODE system are known, can we infer characteristics of the original system $\dot{x}=f(x)$? 2. If $f(x)$ is unknown and only data generated by $f(x)$ is available, can we design a controller based on the NODE dynamics to ensure desirable system properties, such as convergence? Although extensive research has analyzed the properties of NN, including safety verification and robustness \cite{fazlyab2020safety}, the properties of NODE system have been relatively underexplored. In \cite{kang2021stable}, the author analyze the stability of the NODE with respect to a specific equilibrium. Regarding the second question, some recent works provide relevant insights. In \cite{samanipour2023stability}, the authors conduct a stability analysis and develop controllers for continuous-time dynamical systems where the vector field is modeled using an NN, obtained via system identification techniques. Similarly, \cite{zhou2022neural} proposes a learning framework that stabilizes an unknown nonlinear system using an NN controller while simultaneously learning a neural Lyapunov function to certify the region of attraction of the closed-loop system. 
Additionally, \cite{wang2023model} introduces a semidefinite data-driven program for rigorously verifying the stability and safety of an NN-controlled discrete-time linear system with unknown dynamics. Furthermore, \cite{chen2023adaptive} presents homotopy-based PI algorithms for both linear and nonlinear unknown systems, aiming to derive a stabilizing control policy through initialization. A common assumption in these works is that the system's equilibrium is known and located at the origin. However, when $f(x)$ is unknown,
determining the equilibrium becomes a significant challenge. In such cases, traditional control theory methods may no longer be directly applicable.

Contraction analysis \cite{lohmiller1998contraction} provides a powerful framework for studying the stability and convergence properties of nonlinear systems. Unlike traditional stability methods that focus on equilibrium points or particular trajectories, contraction theory investigates whether all system trajectories converge exponentially toward each other. This characteristic enables the analysis of the system's convergence properties without requiring knowledge of specific equilibrium values or trajectories.
The study of contractive systems has gained popularity in recent years, primarily due to their advantageous properties, including robustness to disturbances and incremental input-to-state stability \cite{angeli2002lyapunov}. Numerous methods have been developed in the literature to analyze the contractivity of nonlinear systems, such as those presented in \cite{aminzare2014contraction,barabanov2019contraction,forni2013differential,lohmiller1998contraction,hu2024enforcing,eising2023cautious,jouffroy2010tutorial}. A system is contractive if and only if the corresponding variational system has uniform global exponential stability, as shown in \cite{barabanov2019contraction,jouffroy2010tutorial}. In \cite{lohmiller1998contraction}, it is demonstrated that the contraction property is guaranteed when the largest eigenvalue of the symmetric part of the matrix in the variational system remains uniformly strictly negative. This condition can be interpreted as requiring the matrix measure of the Jacobian, with respect to the 2-norm, to be strictly negative, as noted in \cite{aminzare2014contraction}. The Finsler-Lyapunov functions were introduced in \cite{forni2013differential} as a tool to assess the exponential stability of the variational system, specifically the incremental exponential stability of the system. 

Data have been used to assess the contractivity of systems, as in \cite{eising2023cautious}, or to enforce contractivity by data-based feedback design, as in \cite{hu2024enforcing}.

Recently, two key approaches have emerged linking NN with contraction analysis. The first focuses on analyzing the contractivity of NN. 
In \cite{centorrino2023euclidean,davydov2024non}, the authors investigate the contractivity of various continuous-time NN, including Hopfield, firing rate, Persidskii, and Lur'e networks, using Euclidean norms or non-Euclidean $l_1/l_\infty$ norms. Similarly, in \cite{centorrino2024modeling}, the contractivity of Hopfield NNs, modeled with dynamic recurrent connections undergoing Hebbian learning, is also analyzed. The second approach involves using NN to approximate the contraction metric of the system and designing controllers based on NN. For example, in \cite{sun2021learning}, a deep NN is trained to learn a control contraction metric, which is subsequently used to design a controller to solve output tracking problems. When approximation errors are present, this controller guarantees probabilistic convergence to trajectories that are close to the reference trajectory. In \cite{tsukamoto2021contraction}, the authors introduce a neural contraction metric specifically designed for contractive systems, parameterized through a recurrent NN. An NN-based control strategy is proposed to address the trajectory tracking problem. The steady-state upper bound of the difference between the target and the learned trajectories is directly proportional to the approximation error. The paper \cite{d2023incremental} proposes a novel sufficient condition for ensuring incremental input-to-state stability in a specific class of recurrent NN. This condition is useful for the design of recurrent NN-based control systems that maintain incremental input-to-state stability. It results in a Linear Matrix Inequality (LMI) constraint for certain Recurrent NN, while not considering the approximation error.

In this paper, we present a novel approach for designing NN-based controllers for continuous-time nonlinear affine systems with unknown drift dynamics. Our approach leverages the principles of contraction analysis to achieve effective control. The approach begins by constructing a Neural Ordinary Differential Equation (NODE) \cite{chen2018neural} system to approximate the autonomous system, where the unknown drift dynamics are estimated using a feedforward NN. A critical challenge in employing NN for system analysis and control lies in ensuring that the approximated systems remain contractive. To address this, we impose constraints on the activation functions using incremental sector bounds and leverage contraction analysis to validate the NN weights. This leads to the derivation of non-convex conditions, providing criteria under which the NODE system is guaranteed to be contractive. To enhance practicality and simplify implementation, we transform the non-convex conditions into a convex form. When using a single-hidden-layer NN, the NODE system can be viewed as a Hopfield NN in continuous time \cite{davydov2024non}. 
A key finding of this study is that ensuring the contractivity of the NODE system guarantees the convergence of the trajectories of the original autonomous system to a neighborhood around the unknown equilibrium point. The size of this neighborhood is explicitly characterized by the approximation error of the NODE model. If the constructed NODE system is not contractive, we propose an NN-based controller for the NODE system. Inspired by the ``Nonlinearity Cancellation" method \cite{de2023learning}, to enforce contractivity within the system. The controller comprises two components: a linear control element, which enforces contraction through adjustable control gains, and an NN component, which compensates for the nonlinearities of the NODE system. This integrated control structure not only maintains the contractivity of the NODE system but also ensures that the trajectories of the nonlinear affine system converge to a specified neighborhood of the unknown equilibrium point. Additionally, to minimize the size of this neighborhood, we formulate an optimization problem incorporating a convex LMI constraint.

The paper is organized as follows. In Section \ref{secppf}, we present preliminaries and the formulation of the problem. Our main results are proposed in Section \ref{secmr}, where we present a sufficient condition for the contractivity of NODE systems and the design of an NN-based controller.  The numerical simulations are provided
in Section \ref{secss} and the conclusions are given in Section \ref{secc}.

\section{Preliminaries and problem formulation}\label{secppf}
\emph{Notation.} The symbols $\mathbb{R}$, $\mathbb{R}_{+}$ denote the set of real nonnegative real numbers, respectively. $\mathbb{R}^{n}$ denotes the $n$ dimensional Euclidean space. We denote the identity matrix with appropriate dimension by $I$. Given a matrix $A=[a_{ij}]$, the symbols $A^{\top}$ refer to the transpose of $A$. For a square matrix $A$, $\lambda(A)$ refers to the set of eigenvalues of $A$. For symmetric metrics $B$ and $C$, $B>0$ ($B\geq 0$) indicates that $B$ is positive definite (positive semidefinite) and $B<0$ ($B\leq 0$) indicates that $B$ is negative definite (negative semidefinite), $B< C$ ($B\leq C$) means $B-C<0$ ($B-C\leq 0$). $\mathrm{BD}$ represents the $\mathrm{blockdiag}$. The norm $\|\cdot\|$ represents some vector norm, or the corresponding induced matrix norm. Finally, whenever it is clear from the context, the symbol ``$\ast$'' inside a matrix stands for the symmetric elements in a symmetric matrix.

Consider the following class of nonlinear affine system
 \begin{equation}\label{ND0}
\dot{x}(t)=f(x)+gu(x),
\end{equation}
where $x(t)\in \mathbb{R}^{m_{k+1}}$ is the system state, $u:\mathbb{R}^{m_{k+1}}\rightarrow\mathbb{R}^l$ is the control input, $f:\mathbb{R}^{m_{k+1}}\rightarrow \mathbb{R}^{m_{k+1}}$ is the drift dynamics, $g\in\mathbb{R}^{m_{k+1}\times l}$ denotes the control effectiveness constant matrix. We assume that the drift dynamics $f(x)$ is unknown \cite{greene2022approximate,greene2023deep}. For control purposes \cite{de2023learning, martinelli2023unconstrained}, we learn $f(x)$ from a set $\mathcal{D}$ by combining a linear component with a feedforward NN. This is accomplished by setting $u=0$ and solving the following optimization problem:
\begin{align}
\underset{A, W^i, b^i}{\text{minimize}}  \quad & \frac{1}{N}\sum_{k=1}^{N} \|\dot{x}(t_k)-Ax(t_k)-Z_{nn}(x(t_k))\|,\label{aca}
\end{align}
where $N$ is the number of the data used for training in $\mathcal{D}$, and $t_k$ denotes the corresponding sampling time, $A\in \mathbb{R}^{m_{k+1}\times m_{k+1}}$ is a constant matrix, $Z_{nn}(x):\mathbb{R}^{m_{k+1}}\rightarrow \mathbb{R}^{m_{k+1}}$ is an NN as the following form
\begin{equation}\label{NN}
 \begin{aligned}
&w^0=x,
\\&w^{i+1}=\phi^{i+1}(W^{i+1}w^{i}+b^{i+1}), i=0,...,k-1,
\\&Z_{nn}(x)=W^{k+1}w^{k}+b^{k+1},
\end{aligned}
\end{equation}
where $m_{k+1}=m_0$, $w^0\in \mathbb{R}^{m_{k+1}}$ is the input to the network, $w^{i}\in\mathbb{R}^{m_{i}}$, $W^{i+1}\in \mathbb{R}^{{m_{i+1}}\times m_{i}}$, $b^{i+1}\in \mathbb{R}^{m_{i+1}}$ for are the output, weight matrix, and bias vector of the $i$th layer, respectively. The nonlinear activation function $\phi^{i+1}:\mathbb{R}^{m_{i+1}}\rightarrow \mathbb{R}^{m_{i+1}}$ is of the form
\begin{equation}\label{NNa}
 \begin{aligned}
\phi^{i+1}(v):=[\varphi^{m_1}(v_{m_1}),\cdots,\varphi^{m_{i+1}}(v_{m_{i+1}})]^{\top},
\end{aligned}
\end{equation}
where $\varphi^{m_i}: \mathbb{R}\rightarrow \mathbb{R}$ is the activation function of each neuron. We focus on differentiable activation functions, such as the sigmoid and tanh functions. Subsequently, we derive the incremental sector bound for the activation functions, which serves as a fundamental component supporting our main results.
\begin{definition}\label{definition1}[Incremental sector bound]
Let $\alpha^i, \beta^i\in \mathbb{R}$ with $\alpha^i\leq \beta^i$. The activation function $\varphi^i$ is incrementally sector bounded by $\bbm{\alpha^i,&\beta^i}$ in the set $\mathcal{D}$ if
\begin{equation}\label{ISB}
 \begin{aligned}
&\Big(\varphi^i(x)-\varphi^i(y)-\alpha^i(x-y)\Big) \Big(\varphi^i(x)-\varphi^i(y)-\beta^i(x-y)\Big)\\&\leq 0, \   \forall x, y\in \mathcal{D}.
\end{aligned}
\end{equation}
\end{definition}

In contrast to the offset local sector bound defined in \cite[Def. 4]{yin2021stability}, the incremental sector bound focuses on the relationship between the activation function values at two distinct states, $x$ and $y$, rather than comparing the state $x$ to a specific point $y^*$. Since the equilibrium is unknown in our case, the offset local sector bound is not suitable for our approach. In Proposition \ref{proposition3}, we demonstrate that the incremental sector bound facilitates ensuring the convergence of a nonlinear system to an unknown equilibrium. Many activation functions, e.g. the sigmoid and tanh functions, are constrained by an incremental bound within the set $\mathcal{D}$. 
The incremental sector bound can be determined as follows: due to $\varphi^{i}(\cdot)$ is a scalar function, applying the mean value theorem, we obtain $\varphi^i(x)-\varphi^i(y)=\frac{\partial \varphi^i}{\partial x} (y^*)(x-y)$, where $y^*\in \bbm{x,&y }$. Subsequently, the incremental sector bound is derived from the upper and lower bounds of $\frac{\partial \varphi^i}{\partial x} (y^*)$, for all $y^*\in\mathcal{D}$. In this paper, we specifically consider the tanh function as the activation function for the NN. To incorporate the incremental sector bound of $\varphi^{i}(\cdot)$ into $\phi(\cdot)$, we present the following fact.
\begin{fact}
 \label{fact1}[Incremental sector bound for vector]: 
The activation function $\phi^i$ given by \eqref{NNa} with incrementally sector bounded activation functions $\varphi^i$ satisfies
\begin{equation}\label{ISBV}
\begin{aligned}
&\bbm{\phi^i(x)-\phi^i(y)-K_1^i (x-y)}^\top\bbm{\phi^i(x)-\phi^i(y)-K_2^i(x-y)}\\&\leq 0, \  \forall x, y\in \mathcal{D},
\end{aligned}
\end{equation}
where $K_1^i =\mathrm{diag}\{\alpha^i\}$, $K_2^i =\mathrm{diag}\{\beta^i\}$.
\end{fact}

\begin{proof}
By using \eqref{NNa} and \eqref{ISB}, we have
\begin{equation}\label{ISBV1}
 \begin{aligned}
&\bbm{\phi^i(x)-\phi^i(y)-K_{1}^i(x-y)}^\top\bbm{\phi^i(x)-\phi^i(y)-K_{2}^i(x-y)}=\\&\sum_{i=1}^{k}\Big(\varphi^i(x)-\varphi^i(y)-\alpha^i(x-y)\Big) \Big(\varphi^i(x)-\varphi^i(y)-\beta^i(x-y)\Big)\\&\leq 0. 
 \end{aligned}
\end{equation}
The proof is complete.\end{proof}\\
In \cite[Lem. 2]{fazlyab2020safety}, the authors analyze the vector incremental sector bounds under the assumption that all activation functions share the same slope restrictions, i.e., $\alpha^i$ and $\beta^i$ are identical. In contrast, we extend this framework to accommodate activation functions with different slope restrictions.

In the following lemma, we incorporate the incremental sector bound condition from Fact \ref{fact1} for all layers in \eqref{NN}. 

Define 
\begin{equation}\label{NNVW}
 \begin{matrix}
 w_j^\phi =\bbm{
w_j^0 \\\vdots \\
w_j^{k-1}
}, & \text{and} & \bar{w}_j^\phi =\bbm{
w_j^1 \\\vdots \\
w_j^k
},
\end{matrix}
\end{equation}
for $j=1,2$, where, $w_j^{k}$ represents the $k$th layer of NN given in \eqref{NN} with different input $x_j$, i.e., $w_j^0=x_j$.   
\begin{lemma}\label{lemma1}
Assume $\phi^i$ is incremental sector bounded by $[K_1^i, K_2^i]$. Then, the following inequality holds 
\begin{equation}\label{LEM1}
 \begin{aligned}
\begin{bmatrix}
w_1^{\phi}-w_2^{\phi}
 \\
\bar{w}_1^{\phi}-\bar{w}_2^{\phi}
\end{bmatrix}^\top \begin{bmatrix}
 -2Q_1^\top Q_2&Q_1^{\top}+Q_2^\top
 \\
\ast&-2I
\end{bmatrix}\begin{bmatrix}
w_1^{\phi}-w_2^{\phi}
 \\
\bar{w}_1^{\phi}-\bar{w}_2^{\phi}
\end{bmatrix}\geq 0, 
\end{aligned}
\end{equation}
for all $x_j\in \mathcal{D}$, where $Q_1=\mathrm{BD}(K_1^1W^1, \dots, K_1^{k}W^{k})$, $Q_2=\mathrm{BD}(K_2^1W^1, \dots, K_2^{k}W^{k})\in \mathbb{R}^{\sum_{i=0}^{k-1} m_{i+1}\times \sum_{i=0}^{k-1} m_{i}}$.
\end{lemma}

\begin{proof}
Substitute $Q_1$, $Q_2$ into \eqref{LEM1}, and using Fact \ref{fact1}, one has
\begin{equation}\label{LEMP1}
 \begin{aligned}
&\begin{bmatrix}
w_1^{\phi}-w_2^{\phi}
 \\
\bar{w}_1^{\phi}-\bar{w}_2^{\phi}
\end{bmatrix}^\top \begin{bmatrix}
 -2Q_1^\top Q_2&Q_1^\top+Q_2^\top
 \\
\ast&-2I
\end{bmatrix}\begin{bmatrix}
w_1^{\phi}-w_2^{\phi}
 \\
\bar{w}_1^{\phi}-\bar{w}_2^{\phi}
\end{bmatrix}=\\&-2\sum_{i=0}^{k-1}[(w_1^{i+1}-w_2^{i+1})-K_1^{i+1}W^{i+1}(w_1^{i}-w_2^{i})]^\top\times\\&[(w_1^{i+1}-w_2^{i+1})-K_2^{i+1}W^{i+1}(w_1^{i}-w_2^{i})]\geq 0.
\end{aligned}
\end{equation}
The proof is complete.\end{proof}

Based on universal approximation theorem \cite{hornik1989multilayer}, the nonlinear affine system \eqref{ND0} can be rewritten as
 \begin{equation}\label{ND1}
\dot{x}=Ax+Z_{nn}(x)+gu(x)+\varepsilon(x), 
\end{equation}
for all $x\in \mathcal{D}$, where $\varepsilon(x):\mathbb{R}^{m_{k+1}}\rightarrow \mathbb{R}^{m_{k+1}}$ is an approximation error.


We aim to analyze the contractivity of the approximated NODE system, which is given by the following equation:
 \begin{equation}\label{NDL}
\dot{x}_{nn}(t)=Ax_{nn}+Z_{nn}(x_{nn}), 
\end{equation}
and to design a state feedback controller $u(x)$ to ensure the convergence of the system described by \eqref{ND1}, utilizing the contractivity of the NODE system in \eqref{NDL}.

\section{Main result}\label{secmr}
In this section, we assume that there exists a subset $\mathcal{C}\subseteq \mathcal{D}$ such that the NODE system \eqref{NDL} is forward complete and forward invariant on $\mathcal{C}$, i.e., for any $x_{nn}(0)\in \mathcal{C}$, the trajectory of the NODE system \eqref{NDL}, $x_{nn}(t)$ exists and satisfies $x_{nn}(t)\in\mathcal{C}$, for all $t\geq 0$ \cite{strasser2024koopman}. We begin by analyzing the contractivity of \eqref{NDL} in $\mathcal{C}$. Then, based on the contractive property, we proceed to design an NN-based controller for the original system \eqref{ND0} to guarantee the convergence property.

\subsection{Contractive neural ordinary differential equation}\label{CNODE}
We establish conditions on the weights $W^{i+1}$, and bias $b^{i+1}$ in \eqref{NN} to guarantee that the NODE system \eqref{NDL} exhibits the contraction property by using the following definition: 
\begin{definition}\label{d1}
A NODE system \eqref{NDL} is called {\em contractive} in a forward invariant set $\mathcal{C}$, if there exist positive numbers $c$ and $\alpha$ such that for any pair of solutions $x_{nn1}(t)$, $x_{nn2}(t)$ with the initial condition $x_{nn1}(0)$, $x_{nn2}(0)\in \mathcal{C}$, the following inequality holds:
\begin{equation} \label{eq:d1}
\begin{aligned}
\|
x_{nn1}(t)-x_{nn2}(t)
\|\leq ce^{-\alpha t}\|
x_{nn1}(0)-x_{nn2}(0)
\|, \forall t\geq 0.
\end{aligned}
\end{equation}
\end{definition}

For a clearer understanding of this paper, we first introduce the construction of a contractive NODE system using single-layer neural networks, followed by an extension to multi-layer architectures.

\subsubsection{‌Hopfield (Single-layer) neural network case} 
It has been demonstrated in \cite{lu2021learning} that a single-hidden-layer neural network can approximate any continuous nonlinear function with arbitrary precision, given a sufficient number of neurons. This property makes the
single-hidden layer NN ($k=1$) commonly used for system representation in control theory \cite{d2023incremental,kamalapurkar2016model,samanipour2023stability}. 
For a single-hidden layer NN, the NODE system \eqref{NDL} can be considered as a continuous-time Hopfield NN \cite{davydov2024non}, the incremental sector bounds are determined by $W^1x + b^1$. Because $W^1x + b^1$ is convex, its upper and lower bounds can be easily identified on the boundary of the set $\mathcal{C}$, thereby simplifying the computation of the incremental sector bounds. In the following lemma, we derive the contraction conditions for the continuous-time Hopfield NN.

\begin{theorem}\label{CO0}
Suppose there exist a positive definite matrix $P$, positive constants $\underline{p}$, $\overline{p}$, and $\gamma$, such that the following conditions hold: 
\begin{equation}\label{TH110}
 \underline{p}I\leq P\leq \overline{p}I, 
\end{equation}
\begin{equation}\label{TH120}
\begin{aligned}
&\begin{bmatrix}
A^\top P+PA+\gamma P-2\alpha^1\beta^1W^{1\top}W^1 &L\\\ast
 & -2I
\end{bmatrix}\leq 0,
\end{aligned}
\end{equation}
where $L=PW^2+(\alpha^1+\beta^1)W^{1\top}$, $W^1$, $W^2$, $\alpha^1$, $\beta^1$ are given in \eqref{NN}, \eqref{ISB}, respectively. Then, the NODE system \eqref{NDL} is contractive in $\mathcal{C}$.
\end{theorem}
\begin{proof}
  For single-layer NN, we have $w_j^\phi=x_j$, $\bar{w}_j^\phi=w_j^1$, $Q_1=\alpha^1W^1$, and $Q_2=\beta^1W^1$. Then the inequality \eqref{LEM1} can be rewritten as
  \begin{equation}\label{LEM10}
 \begin{aligned}
&\begin{bmatrix}
x_{nn1}-x_{nn2}
 \\
w_{nn1}^1-w_{nn2}^1
\end{bmatrix}^\top \begin{bmatrix}
 -2\alpha^1\beta^1W^{1\top} W^1 &(\alpha^1+\beta^1)W^{1\top}
 \\
\ast&-2I
\end{bmatrix}\\&\times\begin{bmatrix}
x_{nn1}-x_{nn2}
 \\
w_{nn1}^1-w_{nn2}^1
\end{bmatrix}\geq 0, 
\end{aligned}
\end{equation}
By defining the Lyapunov function $V(x_{nn1}-x_{nn2})=\bbm{x_{nn1}-x_{nn2}}^\top P \bbm{x_{nn1}-x_{nn2}}$, and using \eqref{NN}, \eqref{TH120}, \eqref{LEM10}, one has  
\begin{equation}\label{CO0P1}
\begin{aligned}
&\dot{V}(x_{nn1}-x_{nn2})\overset{\eqref{NN}}{=}\\&\bbm{x_{nn1}-x_{nn2}}^\top( A^\top P+PA)\bbm{x_{nn1}-x_{nn2}}+\\&\bbm{w_{nn1}^1-w_{nn2}^1}^\top W^{1\top}P\bbm{x_{nn1}-x_{nn2}}+\\&\bbm{x_{nn1}-x_{nn2}}^\top PW^1\bbm{w_{nn1}^1-w_{nn2}^1}\\&\overset{\eqref{TH120}}{\leq}-\gamma \bbm{x_{nn1}-x_{nn2}}^\top P \bbm{x_{nn1}-x_{nn2}}-\begin{bmatrix}
x_{nn1}-x_{nn2}
 \\
w_{nn1}^1-w_{nn2}^1
\end{bmatrix}^\top \\&\times \begin{bmatrix}
 -2\alpha^1\beta^1W^{1\top} W^1 &(\alpha^1+\beta^1)W^{1\top}
 \\
\ast&-2I
\end{bmatrix}\begin{bmatrix}
x_{nn1}-x_{nn2}
 \\
w_{nn1}^1-w_{nn2}^1
\end{bmatrix}\\&\overset{\eqref{LEM10}}{\leq}  -\gamma \bbm{x_{nn1}-x_{nn2}}^\top P \bbm{x_{nn1}-x_{nn2}}\\&=-\gamma V(x_{nn1}-x_{nn2}).
\end{aligned}
\end{equation}
By using the Comparison Lemma and \eqref{TH110}, we have 
\begin{equation}\label{TH1P3}
\begin{aligned}
\|x_{nn1}(t)-x_{nn2}(t)\|\leq \sqrt{\frac{\overline{p} }{\underline{p} } } e^{-\frac{\gamma}{2} t} \|x_{nn1}(0)-x_{nn2}(0)\|.
\end{aligned}
\end{equation}
The proof is thus completed with $c=\sqrt{\frac{\overline{p} }{\underline{p} } }$, $\alpha=\frac{\gamma}{2}$ in \eqref{eq:d1}.
\end{proof}

The contractive condition \eqref{TH120} is non-convex when optimizing over both $P$ and $\gamma$ simultaneously. non-convex conditions can complicate the search for feasible solutions, as the feasible region may be disjoint or have a complex structure \cite{boyd2004convex}. To address this challenge, we will reformulate the contractive condition \eqref{TH120} into a convex LMI to facilitate a more straightforward optimization process.


\begin{theorem}\label{the11} Suppose there exist a positive definite matrix $P$, positive constants $\underline{p}$, $\overline{p}$, and $\mu$, such that  \eqref{TH110} and the following condition
\begin{equation}\label{CO12}
\begin{bmatrix}
-\mu I&P &0\\\ast&-2\alpha^1\beta^1W^{1\top}W^1+A^\top P+PA&L\\
\ast&\ast& -2I
\end{bmatrix}\leq 0,
\end{equation} 
hold, where $L=PW^2+(\alpha^1+\beta^1)W^{1\top}$, $W^1$, $W^2$, $\alpha^1$, $\beta^1$ are given in \eqref{NN}, \eqref{ISB}, respectively. Then the NODE system \eqref{NDL} is contractive in $\mathcal{C}$.
\end{theorem}

\begin{proof}
By utilizing the Schur complement on \eqref{CO12}, we can express \eqref{CO12} equivalently by   
\begin{equation}\label{COp1}
\begin{aligned}
&\begin{bmatrix}
-\mu I&P\\P&-2\alpha^1\beta^1W^{1\top}W^1+A^\top P+PA
\end{bmatrix}+\frac{1}{2}\begin{bmatrix}
0\\L
\end{bmatrix}\begin{bmatrix}
0&L^\top
\end{bmatrix}\\&=\begin{bmatrix}
-\mu I&P\\P&-2\alpha^1\beta^1W^{1\top}W^1+A^\top P+PA+\frac{1}{2}LL^{\top}
\end{bmatrix}\leq 0.
\end{aligned}
\end{equation}
Applying the Schur complement to \eqref{COp1} once again, we obtain
\begin{equation}\label{COp2}
\begin{aligned}
&-2\alpha^1\beta^1W^{1\top}W^1+A^\top P+PA+\frac{1}{2}LL^{\top}+\frac{1}{\mu}P^2\leq 0.
\end{aligned}
\end{equation}
By choosing $\gamma=\frac{\underline{p}}{\mu}$, we have
\begin{equation}\label{COp3}
\begin{aligned}
&A^\top P+PA+\gamma P-2\alpha^1\beta^1W^{1\top}W^1+\frac{1}{2}LL^{\top}\leq\\&A^\top P+PA+\frac{1}{\mu}P^2-2\alpha^1\beta^1W^{1\top}W^1+\frac{1}{2}LL^{\top}\leq 0.
\end{aligned}
\end{equation}
The inequality \eqref{COp3} is equivalent to the Schur complement of \eqref{TH120}. Since the conditions in Theorem \ref{CO0} are satisfied, its conclusion follows directly.
\end{proof}

\subsubsection{‌Multilayer neural network case}
For the multilayer NN, we present the following theorem outlining the construction of a contractive NODE system.
\begin{theorem}\label{theorem1}
Suppose there exist a positive definite matrix $P$, positive constants $\underline{p}$, $\overline{p}$, and $\gamma$ such that \eqref{TH110} and the following condition 
\begin{equation}\label{TH12}
\begin{aligned}
&U_1^\top
\begin{bmatrix}
     A^\top P+PA+\gamma P &PW^{k+1}\\
  \ast& 0
\end{bmatrix}U_1+\\&U_2^\top\begin{bmatrix}
 -2Q_1^\top Q_2&Q_1^\top+Q_2^\top
 \\
\ast&-2I
\end{bmatrix}U_2\leq 0, \quad\forall k\geq 2,
\end{aligned}
\end{equation}
hold, where $U_1=\bbm{I_{m_{k+1}\times m_{k+1}} &0_{m_{k+1}\times a}&0_{m_{k+1}\times m_{k}}\\0_{m_{k}\times m_{k+1}}&0_{m_{k}\times a}&I_{m_{k}\times m_{k}}}$, $U_2=\bbm{I_{m_{k+1}\times m_{k+1}} &0_{m_{k+1}\times a}&0_{m_{k+1}\times m_{k}}\\0_{a\times m_{k+1}}&I_{a\times a}&0_{a\times m_{k}}\\0_{a\times m_{k+1}}&I_{a\times a}&0_{a\times m_{k}}\\0_{m_{k}\times m_{k+1}}&0_{m_{k}\times a}&I_{m_{k}\times m_{k}}}$, $a=\sum_{i=1}^{k-1}m_i$, $W^{k+1}$, $Q_1$, $Q_2$ are given in \eqref{NN}, \eqref{LEM1}, respectively. Then the NODE system \eqref{NDL} is contractive in $\mathcal{C}$.
\end{theorem}

\begin{proof}
By multiplying the left-hand side of inequality \eqref{TH12} by $\bbm{
w_{nn1}^0-w_{nn2}^0
 \\
\bar{w}_{nn1}^{\phi}-\bar{w}_{nn2}^{\phi}
}^\top$ on the left, $\bbm{
w_{nn1}^0-w_{nn2}^0
 \\
\bar{w}_{nn1}^{\phi}-\bar{w}_{nn2}^{\phi}
}$ on the right, and by using \eqref{TH12}, one has 
\begin{equation}\label{TH1P1}
\begin{aligned}
&\begin{bmatrix}
  x_{nn1}-x_{nn2}\\
w_{nn1}^k -w_{nn2}^k
\end{bmatrix}^\top
\begin{bmatrix}
   A^\top P+PA+ \gamma P &PW^{k+1}\\
   \ast & 0
\end{bmatrix}\times\\&\begin{bmatrix}
   x_{nn1}-x_{nn2}\\
w_{nn1}^k -w_{nn2}^k
\end{bmatrix}+\begin{bmatrix}
w_{nn1}^{\phi}-w_{nn2}^{\phi}
 \\
\bar{w}_{nn1}^{\phi}-\bar{w}_{nn2}^{\phi}
\end{bmatrix}^\top \times\\&\begin{bmatrix}
 -2Q_1^\top Q_2&Q_1^\top+Q_2^\top
 \\
\ast&-2I
\end{bmatrix}\begin{bmatrix}
w_{nn1}^{\phi}-w_{nn2}^{\phi}
 \\
\bar{w}_{nn1}^{\phi}-\bar{w}_{nn2}^{\phi}
\end{bmatrix}\leq 0.
\end{aligned}
\end{equation}
Define the Lyapunov function $V( x_{nn1}-x_{nn2})=\bbm{ x_{nn1}-x_{nn2}}^\top P \bbm{ x_{nn1}-x_{nn2}}$, use \eqref{NN}, \eqref{TH1P1}, \eqref{LEM1}, one has
\begin{equation}\label{TH1P2}
\begin{aligned}
&\dot{V}( x_{nn1}-x_{nn2})\overset{\eqref{NN}}{=}\\&\bbm{ x_{nn1}-x_{nn2}}^\top( A^\top P+PA)\bbm{ x_{nn1}-x_{nn2}}+\\&\bbm{w_{nn1}^k-w_{nn2}^k}^\top  W^{k\top}P\bbm{ x_{nn1}-x_{nn2}}+\\&\bbm{ x_{nn1}-x_{nn2}}^\top PW^k\bbm{w_{nn1}^k-w_{nn2}^k}\\&\overset{\eqref{TH1P1}}{\leq}-\gamma \bbm{ x_{nn1}-x_{nn2}}^\top P \bbm{ x_{nn1}-x_{nn2}}-\\&\begin{bmatrix}
w_{nn1}^{\phi}-w_{nn2}^{\phi}
 \\
\bar{w}_1^{\phi}-\bar{w}_2^{\phi}
\end{bmatrix}^\top \begin{bmatrix}
 -2Q_1^\top Q_2&Q_1^\top+Q_2^\top
 \\
\ast&-2I
\end{bmatrix}\begin{bmatrix}
w_{nn1}^{\phi}-w_{nn2}^{\phi}
 \\
\bar{w}_{nn1}^{\phi}-\bar{w}_{nn2}^{\phi}
\end{bmatrix}\\&\overset{\eqref{LEM1}}{\leq}  -\gamma \bbm{x_{nn1}-x_{nn2}}^\top P \bbm{x_{nn1}-x_{nn2}}\\&=-\gamma V(x_{nn1}-x_{nn2}).
\end{aligned}
\end{equation}
The proof is completed by using the comparison lemma.
 \end{proof}
 
\begin{remark}\label{REM0}
Theorem \ref{theorem1} is designed to impose restrictions specifically on the weights of the NN and the incremental sector bound, without restricting the biases of the NN. This approach has an important implication. While the theorem does not directly address the biases, it does not ignore them. In fact, the biases are implicitly captured within the incrementally sector bound itself. As shown in Lemma \ref{lemma1}, the matrices $Q_1$ and $Q_2$ illustrate that the incremental sector bounds $K_1$ and $K_2$ are intrinsically determined by $W^{i+1}$, $w^i$, and $b^{i+1}$. Consequently, although the theorem imposes explicit conditions only on the weights, the role of biases is still present in the overall structure, ensuring that they are indirectly accounted for within the constraints of the incrementally sector bound.  
\end{remark}

The non-convex contractive condition for NODE is also considered in \cite{martinelli2023unconstrained}. However, the feasibility of the LMI in \cite[Thm.1]{martinelli2023unconstrained} relies on the matrix $A$ being Hurwitz. In Theorem \ref{theorem1}, we overcome this requirement by introducing the term $-2Q_1^\top Q_2$. Notably, in Example \ref{example2}, the matrix $A$ is non-Hurwitz, with eigenvalues $0.0006\pm1.3284j$. The next theorem reformulate the contractive condition \eqref{TH12} into a convex LMI.

\begin{theorem}\label{theorem11}
Suppose there exist a positive definite matrix $P$, positive constants $\underline{p}$, $\overline{p}$, and $\mu$, such that \eqref{TH110} and the following condition
\begin{equation}\label{TH121}
\begin{aligned}
&\bar{U}_1^\top
\bbm{-\mu I&P&0\\
   \ast&  A^\top P+PA &PW^{k+1}\\
  \ast& \ast & 0
}\bar{U}_1+\\&\bar{U}_2^\top\bbm{0&0&0\\\ast&
 -2Q_1^\top Q_2&Q_1^\top+Q_2^\top
 \\\ast&
\ast&-2I
}\bar{U}_2\leq 0, \quad\forall k\geq 2,
\end{aligned}
\end{equation}
hold, where \\$\bar{U}_1=\bbm{I_{m_{k+1}\times m_{k+1}} &0_{m_{k+1}\times m_{k+1}}&0_{m_{k+1}\times a} &0_{m_{k+1}\times m_{k}}\\0_{m_{k+1}\times m_{k+1}}&I_{m_{k+1}\times m_{k+1}}&0_{m_{k+1}\times a}&0_{m_{k+1}\times m_{k}}\\0_{m_{k}\times m_{k+1}}&0_{m_{k}\times m_{k+1}}&0_{m_{k}\times a}&I_{m_{k}\times m_{k}}}$, $\bar{U}_2=\bbm{I_{m_{k+1}\times m_{k+1}} &0_{m_{k+1}\times m_{k+1}}&0_{m_{k+1}\times a}&0_{m_{k+1}\times m_{k}}\\0_{m_{k+1}\times m_{k+1}}&I_{m_{k+1}\times m_{k+1}} &0_{m_{k+1}\times a}&0_{m_{k+1}\times m_{k}}\\0_{a\times m_{k+1}}&0_{a\times m_{k+1}}&I_{a\times a}&0_{a\times m_{k}}\\0_{a\times m_{k+1}}&0_{a\times m_{k+1}}&I_{a\times a}&0_{a\times m_{k}}\\0_{m_{k}\times m_{k+1}}&0_{m_{k}\times m_{k+1}}&0_{m_{k}\times a}&I_{m_{k}\times m_{k}}}$, $a=\sum_{i=1}^{k-1}m_i$, $W^{k+1}$, $Q_1$, $Q_2$ are given in \eqref{NN}, \eqref{LEM1}, respectively. Then the NODE system \eqref{NDL} is contractive in $\mathcal{C}$.
\end{theorem}
\begin{proof}
 Since $\bar{U}_1=\bbm{I_{m_{k+1}\times m_{k+1}}&0\\0&U_1}$, $\bar{U}_2=\bbm{I_{m_{k+1}\times m_{k+1}}&0\\0&U_2}$, we have
 \begin{equation}\label{TH121P1}
\begin{aligned}
&\bar{U}_1^\top
\bbm{-\mu I&P&0\\
   \ast&  A^\top P+PA &PW^{k+1}\\
  \ast& \ast & 0
}\bar{U}_1\\&=\bbm{I&0\\0&U^\top_1}\bbm{-\mu I&P&0\\
   \ast&  A^\top P+PA &PW^{k+1}\\
  \ast& \ast & 0
}\bbm{I&0\\0&U_1}\\&=
\bbm{-\mu I&\bbm{P&0}\\U_1^\top \bbm{P\\0}&U_1^\top\bbm{A^\top P+PA &PW^{k+1}\\
 \ast& 0}}\bbm{I_{n\times n}&0\\0&U_1}\\&=\bbm{-\mu I&\bbm{P&0}U_1\\\ast&U_1^\top\bbm{A^\top P+PA &PW^{k+1}\\
  \ast & 0}U_1}.
\end{aligned}
\end{equation}
Similarly, we have
\begin{equation}\label{TH121P2}
\begin{aligned}
&\bar{U}_2^\top\bbm{0&0&0\\\ast&
 -2Q_1^\top Q_2&Q_1^\top+Q_2^\top
 \\\ast&
\ast&-2I
}\bar{U}_2\\&=\bbm{0&0\\ 0&U_2^\top\bbm{-2Q_1^\top Q_2&Q_1^\top+Q_2^\top
 \\
\ast&-2I}U_2}.
\end{aligned}
\end{equation}
Then inequality \eqref{TH121} can be written as
\begin{equation}\label{TH121P3}
\begin{aligned}
&\bbm{-\mu I&\bbm{P&0}U_1\\\ast&\bar{L}}\leq 0,
\end{aligned}
\end{equation}
where $\bar{L}=U_1^\top
\begin{bmatrix}
     A^\top P+PA &PW^{k+1}\\
   \ast& 0
\end{bmatrix}U_1+U_2^\top\begin{bmatrix}
 -2Q_1^\top Q_2&Q_1^\top+Q_2^\top
 \\
\ast&-2I
\end{bmatrix}U_2$. By applying the Schur complement to inequality \eqref{TH121P3}, one has
\begin{equation}\label{TH121P4}
\begin{aligned}
&U_1^\top\bbm{\frac{1}{\mu}P^2&0\\0&0}U_1+\bar{L}=\\&U_1^\top
\begin{bmatrix}
     A^\top P+PA+\frac{1}{\mu}P^2 &PW^{k+1}\\
   \ast & 0
\end{bmatrix}U_1+\\&U_2^\top\begin{bmatrix}
 -2Q_1^\top Q_2&Q_1^\top+Q_2^\top
 \\
\ast&-2I
\end{bmatrix}U_2\leq 0.
\end{aligned}
\end{equation}
By choosing $\gamma=\frac{\underline{p}}{\mu}$, we have
\begin{equation}\label{THp121P5}
\begin{aligned}
&U_1^\top\bbm{\gamma P&0\\0&0}U_1=\bbm{\gamma P&0&0\\\ast&0&0\\\ast&\ast&0}\leq \bbm{\frac{1}{\mu}P^2&0&0\\\ast&0&0\\\ast&\ast&0}\\&=U_1^\top\bbm{\frac{1}{\mu}P^2&0\\0&0}U_1.
\end{aligned}
\end{equation}
Then inequality \eqref{TH12} is satisfied by applying inequality \eqref{TH121P4}, and \eqref{THp121P5}. The conditions in Theorem \ref{theorem1} are met, and the
conclusion of the theorem follows directly. 
\end{proof}

The next proposition provides bounds on the distance between trajectory of the autonomous system
\begin{equation}\label{ND}
\dot{x}=Ax+Z_{nn}(x)+\varepsilon(x)=:f(x), 
\end{equation}
and the trajectory of the contractive NODE system \eqref{NDL} in $\mathcal{C}$.
\begin{proposition}\label{proposition2}
Consider the NODE system \eqref{NDL} with initial condition $x_{nn}(0)\in \mathcal{C}$, and assume that conditions \eqref{TH110}, \eqref{TH12} are satisfied. Likewise, consider system \eqref{ND} with initial condition $x(0)\in\mathcal{C}$, and let $T$ be such that $x(t)\in \mathcal{C}$ for all $t\in [0,\ T)$. Then for any $t\in [0,\ T)$, the trajectory of the NODE system \eqref{NDL} and the trajectory of system \eqref{ND} satisfy the following inequality:
 \begin{equation}\label{PR21}
 \begin{aligned}
&\|x(t)-x_{nn}(t)\|\leq \\&\sqrt{\frac{\overline{p}}{\underline{p}}}\|x(0)-x_{nn}(0)\|e^{-\frac{\gamma }{2}t}+\frac{2\epsilon}{\gamma}\sqrt{\frac{\overline{p}}{\underline{p}}}(1-e^{-\frac{\gamma }{2}t}). 
\end{aligned}
\end{equation}
Here, $\overline{p}$, $\underline{p}$, $\gamma$ are defined in \eqref{TH110}, \eqref{TH12}, $\epsilon=\sup \|\varepsilon(x)\|$.
\end{proposition}
\begin{proof}
Define the Lyapunov function $V(x-x_{nn})=\bbm{x-x_{nn}}^\top P \bbm{x-x_{nn}}$, $P$ is the same as in Theorem \ref{theorem1}. The derivative of $V(x-x_{nn})$ along the 
trajectories of system \eqref{NDL} and \eqref{ND} satisfies
 \begin{equation}\label{PR2p1}
 \begin{aligned}
&\dot{V}(x-x_{nn})=\bbm{x-x_{nn}}^\top( A^\top P+PA)\bbm{x-x_{nn}}+\\&\bbm{w^k-w_{nn}^k}^\top W^{(k+1)\top}P\bbm{x-x_{nn}}+\\&\bbm{x-x_{nn}}^\top PW^{k+1}\bbm{w^k-w_{nn}^k}+2\varepsilon(x)^\top P\bbm{x-x_{nn}}\\&\overset{\eqref{TH1P2}}{\leq}-\gamma V(x-x_{nn})+2\|\varepsilon(x)\| \|P\bbm{x-x_{nn}}\|\\&\overset{\eqref{TH110}}{\leq}-\gamma V(x-x_{nn})+2\epsilon\sqrt{\overline{p}}\sqrt{V(x-x_{nn})},
\end{aligned}
\end{equation}
By the comparison lemma, $V(x-x_{nn})$ satisfies the inequality 
 \begin{equation}\label{PR2p2}
 \begin{aligned}
\sqrt{V(x(t)-x_{nn}(t))}\leq e^{-\frac{\gamma }{2}t}\sqrt{V_0}+\frac{2\epsilon\sqrt{\overline{p}}}{\gamma}(1-e^{-\frac{\gamma }{2}t}).
\end{aligned}
\end{equation}
By using \eqref{TH110}, we arrive at 
\begin{equation}\label{PR2p3}
 \begin{aligned}
&\sqrt{\underline{p}}\|x(t)-x_{nn}(t)\|\\&\leq e^{-\frac{\gamma }{2}t}\sqrt{\overline{p}}\|x(0)-x_{nn}(0)\|+\frac{2\epsilon\sqrt{\overline{p}}}{\gamma}(1-e^{-\frac{\gamma }{2}t}).
\end{aligned}
\end{equation}
This gives the desired result.
\end{proof}


For system \eqref{ND}, assume that there exists a subset $\mathcal{B}\subseteq\mathcal{C}$ such that, for any initial condition $x(0)\in \mathcal{B}$, the trajectory $x(t)$ remains within $\mathcal{C}$ for all $t\geq 0$. A similar assumption can also be found in \cite{strasser2024koopman,van2023reprojection}. Moreover, suppose there exists an unknown equilibrium point $x^*\in \mathcal{C}$, i.e., there exists an unknown $x^*\in \mathcal{C}$ such that $Ax^*+Z_{nn}(x^*)+\varepsilon(x^*)=0$. The following proposition establishes that if the NODE system \eqref{NDL} is contractive, the nonlinear system \eqref{ND} converges to a neighborhood of its equilibrium point, with the size of the neighborhood determined by the approximation error $\varepsilon(x)$.
\begin{proposition}\label{proposition3}
Consider the nonlinear system \eqref{ND} with the initial condition $x(0)\in \mathcal{B}$. If the NODE system \eqref{NDL} satisfy conditions \eqref{TH110}, \eqref{TH12}, then the trajectory of system \eqref{ND} satisfies the following inequality:
 \begin{equation}\label{pro3}
 \begin{aligned}
&\|x(t)-x^*\|\leq \sqrt{\frac{\overline{p}}{\underline{p}}}\|x(0)-x^*\|e^{-\frac{\gamma }{2}t}+\frac{4\epsilon}{\gamma}\sqrt{\frac{\overline{p}}{\underline{p}}}(1-e^{-\frac{\gamma }{2}t}), 
\end{aligned}
\end{equation} 
for all $t\geq0$, where $x^*$ is the unknown equilibrium of system \eqref{ND} in $\mathcal{C}$.
\end{proposition}

\begin{proof}
Consider the NODE system \eqref{NDL} with an external input $\varepsilon(x^*)$, which is given by:
 \begin{equation}\label{NDLS}
\dot{x}_{nn}^{\varepsilon}(t)=Ax_{nn}^{\varepsilon}+Z_{nn}(x_{nn}^{\varepsilon})+\varepsilon(x^*).
\end{equation}
Here, $x^*$ represents the equilibrium of the system \eqref{ND}, satisfying $0=Ax^*+Z_{nn}(x^*)+\varepsilon(x^*)$. Consequently, $x^*$ also serves as a trajectory for \eqref{NDLS}. Given that the conditions \eqref{TH110}, \eqref{TH12} are fulfilled, the NODE system \eqref{NDLS} is contractive by Theorem \ref{theorem1}. Following the proof of Proposition \ref{proposition2}, replacing $x_{nn}$ with $x^*$, we define the Lyapunov function as $V(x-x^*)=\bbm{x-x^*}^\top P \bbm{x-x^*}$, $P$ is the same as in Proposition \ref{proposition2}. Taking the derivative of $V(x-x^*)$ along the trajectory of the systems \eqref{ND}, \eqref{NDLS}, and applying \eqref{TH1P2}, we have:
\begin{equation}\label{TH2P3}
\begin{aligned}
&\dot{V}(x-x^*)
=\bbm{x-x^*}^\top( A^\top P+PA)\bbm{x-x^*}+\\&\bbm{w_1^k-w_{2}^k}^\top W^{(k+1)\top}P\bbm{x-x^*}+\\&\bbm{x-x^*}^\top PW^{k+1}\bbm{w^k-w_{nn}^k}+\\&2\bbm{\varepsilon(x)-\varepsilon(x^*)}^\top P\bbm{x-x_{nn}^{\varepsilon}}\\&\overset{\eqref{TH1P2}}{\leq}-\gamma \bbm{x-x^*}^\top P \bbm{x-x^*}+\\&2\|x-x^*\|\|P\|\|\varepsilon(x)-\varepsilon(x^*)\|\\&\overset{\eqref{TH110}}{\leq}-\gamma V(x-x^*)+4\epsilon\sqrt{\overline{p}}\sqrt{V(x-x^*)}.
\end{aligned}
\end{equation}
The rest of the proof follows \eqref{PR2p1}-\eqref{PR2p3}. 
\end{proof}

In the proof of Proposition \ref{proposition3}, we only use the contractivity of \eqref{NDLS}, which is independent of $\varepsilon(x^*)$. Consequently, knowing the exact value of $x^*$ is not required. From Proposition \ref{proposition3}, it follows that for a given approximation error $\varepsilon(x)$, the trajectory of system \eqref{ND} asymptotically converges to a neighborhood of the unknown equilibrium point $x^*$ as $t\to \infty$. Specifically, the state $x(t)$ lies within the ball $\|x-x^*\|\leq \frac{4\epsilon}{\gamma}\sqrt{\frac{\overline{p}}{\underline{p}}}$, whose size is determined by the approximation error.

\subsection{Neural network based controller design}
In Subsection \ref{CNODE}, we established the contractivity of the NODE system \eqref{NDL}. Here, we consider the scenario where the NODE system \eqref{NDL} is not contractive. Our goal is to design an NN-based controller that enforces contractivity in the NODE system \eqref{NDL}, utilizing the results presented in Subsection \ref{CNODE}. By applying Proposition \ref{proposition3}, the proposed controller guarantees that the trajectories of the original system \eqref{ND0} converge to a neighborhood of the unknown equilibrium point $x^*$.

Consider the NODE affine system given by the following form
\begin{equation}\label{NNC}
  \dot{x}=Ax+Z_{nn}(x)+gu(x),
\end{equation}
where the system is defined in $\mathcal{C}$, $g$ is a known constant matrix, and $Z_{nn}(x)$ is defined in \eqref{NN}. The controller $u(x)$ consists of a linear state feedback controller $Hx$ and an NN controller $u_{nn}(x)$, i.e., 
\begin{equation}\label{NNC10}
u(x)=Hx+u_{nn}(x),
\end{equation}
where $H$ represents the control gain and $u_{nn}(x)$ is defined by \eqref{NN} with 
\begin{equation}\label{NNC1}
u_{nn}(x)=W_u^{k+1}w^k+b_u^{k+1}. 
\end{equation}
This indicates that the NN controller $u_{nn}(x)$ shares the same structure as $Z_{nn}(x)$, differing only in the outer weights and biases. The controller design follows the ``Nonlinearity Cancellation" approach \cite{de2023learning}, where $Z_{nn}(x)$ can be considered as the system's basis functions. Since $Z_{nn}(x)$ is bounded both above and below, its effect on the system can be effectively suppressed by the control term $u_{nn}(x)$, thereby enabling the linear component $Ax$ to dominate the system dynamics. The additional state feedback $Hx$ is then introduced to guarantee the contractive property of the overall system.
By using $u(x)=Hx+u_{nn}(x)$, the affine system \eqref{ND0} can be rewritten as  
\begin{equation}\label{NNC2}
  \dot{x}=(A+gH)x+(W^{k+1}+gW_u^{k+1})w^k+b^{k+1}+gb_u^{k+1}+\varepsilon(x).
\end{equation}
Consider the NODE affine system
\begin{equation}\label{NNC3}
  \dot{x}=(A+gH)x+(W^{k+1}+gW_u^{k+1})w^k+b^{k+1}+gb_u^{k+1}.
\end{equation}
Let 
\begin{equation}\label{NNC4}
W_{u\min}^{k+1}=\underset{W_u^{k+1}}{\text{argmin}} \|W^{k+1}+gW_u^{k+1}\|,
\end{equation} 
and define 
\begin{equation}\label{NNC5}
W_{\min}=W^{k+1}+gW_{u\min}^{k+1}. 
\end{equation}
Then, our controller \eqref{NNC10} is in the following form
\begin{equation}\label{NNC11}
u(x)=Hx+W_{u\min}^{k+1}w^k+b_u^{k+1}.
\end{equation}
\subsubsection{‌‌Hopfield (Single-layer) neural network case} 
For single-hidden layer NN, in the light of Lemma \ref{CO0}, we can derive the following result.
\begin{lemma}\label{corollary2}
Suppose there exist a positive definite matrix $P$, positive constants $\underline{p}$, $\overline{p}$, $\gamma$, and a matrix $H$, such that \eqref{TH110} and the following condition 
\begin{equation}\label{CO22}
\begin{aligned}
&\begin{bmatrix}
L'&PW_{\min}+(\alpha^1+\beta^1)W^{1\top}\\
\ast & -2I
\end{bmatrix}\leq 0,
\end{aligned}
\end{equation}
hold, where $W_{\min}=W^2+gW_{u\min}^{2}$, $L'=(A+gH)^\top P+P(A+gH)+\gamma P-2\alpha^1\beta^1W^{1\top}W^1$, and $W^1$, $W^2$, $\alpha^1$, $\beta^1$ are given in \eqref{NN}, \eqref{ISB}, respectively. Then, the learned affine system \eqref{NNC3} is contractive in $\mathcal{C}$ under the controller specified by $u(x)=Hx+W_{u\min}^{2}w^1+b_u^{2}$ with $W_{u\min}^{2}=\underset{W_u^{2}}{\text{argmin}} \|W^{2}+gW_u^{2}\|$.
\end{lemma} 

\begin{remark}\label{rem3}
According to Proposition \ref{proposition3} and Lemma \ref{corollary2}, the controller \eqref{NNC11} ensures that the trajectories of the nonlinear affine system \eqref{ND0} converge to a neighborhood of the equilibrium point $x^*$ which is characterized by $\|x-x^*\|\leq \frac{4\epsilon}{\gamma}\sqrt{\frac{\overline{p}}{\underline{p}}}$. Notably, reducing the size of the neighborhood improves the system's convergence property. As a result, the controller design problem can be reformulated as the following optimization problem:
\begin{align}
\underset{P,H,\gamma}{\text{minimize}}  \quad & \frac{1}{\gamma}\sqrt{\frac{\overline{p}}{\underline{p}}}\nonumber\\
\text{s.t.} \quad & \eqref{TH110}, \eqref{CO22}.\label{OPr}
\end{align}
\end{remark}
We will develop a revised version of \eqref{OPr} that is better suited for implementation, as discussed in Remark \ref{rem4}.
\begin{remark}\label{REM4}
Unlike conventional stabilization methods for nonlinear systems, which typically require knowledge of the equilibrium point to design a stabilizing controller, contractive systems offer a unique advantage: they guarantee stability without prior knowledge of the equilibrium. This property is especially valuable in cases where identifying the equilibrium is difficult or when the equilibrium varies. It is important to highlight in Lemma \ref{corollary2} that $b_u^{2}$ can be chosen freely. This flexibility arises because $b_u^{2}$ is independent of the state variable $x$, meaning it does not affect the contractivity of system \eqref{NNC3}. By appropriately selecting $b_u^{2}$, the system can be steered toward any desired equilibrium, a concept frequently found in game theory \cite{alpcan2009control}, and safety control problems \cite{reis2020control}. For instance, setting $b_u^{2} = -(W^2 + gW_u^{2})w^1|_{x=0} - b^2=\phi^1(b^1)-b^2$ shifts the system's equilibrium to zero (see \eqref{NN}). This approach increases the flexibility of system control and ensures stability, even when the equilibrium is not explicitly known.
\end{remark}

Similar to Theorem \ref{the11}, the non-convex condition \eqref{CO22} can be reformulated as a convex LMI. This leads to the following theorem.
\begin{theorem}\label{the21}
Suppose there exist a positive definite matrix $S$, positive constants $\mu$, $\underline{s}$, $\overline{s}$, and a matrix $Y$, such that the following conditions hold:
\begin{equation}\label{CO31}
 \underline{s}I\leq S\leq \overline{s}I, 
\end{equation}
\begin{equation}\label{CO32}
\begin{aligned}
&\begin{bmatrix}
-\mu I&S &0\\\ast&\bar{R}&\tilde{R}\\
\ast&\ast& -2I
\end{bmatrix}\leq 0,
\end{aligned}
\end{equation}
where $\bar{R}=AS+SA^\top+gY+Y^\top g^\top-\frac{2\alpha^1\beta^1}{(\beta^1-\alpha^1)^2}W_{\min}W_{\min}^{\top}$, $\tilde{R}=\frac{\alpha^1+\beta^1}{\beta^1-\alpha^1}W_{\min}+(\beta^1-\alpha^1)SW^{1\top}$, $W_{\min}=W^2+gW_{u\min}^{2}$. Then, the learned affine system \eqref{NNC3} is contractive in $\mathcal{C}$ under the controller specified by $u(x)=Hx+W_{u\min}^{2}w^1+b_u^{2}$, with $W_{u\min}^{2}=\underset{W_u^{2}}{\text{argmin}} \|W^{2}+gW_u^{2}\|$, $H=YS^{-1}$.
\end{theorem}
\begin{proof}
Following a similar approach as in the proof of Theorem \ref{the11}, by applying the Schur complement to the left side of \eqref{CO32} twice, inequality \eqref{CO32} can be equivalently transformed into
\begin{equation}\label{CO3p1}
\begin{aligned}
&\bar{R}+\frac{1}{2}\tilde{R}\tilde{R}^{\top}+\frac{1}{\mu}S^2=
AS+SA^\top+gY+Y^\top g^\top+\\&\frac{1}{\mu}S^2-2\alpha^1\beta^1SW^{1\top}W^{1}S+\frac{1}{2}(W_{\min}+(\alpha^1+\beta^1)SW^{1\top})\\&\times(W_{\min}+(\alpha^1+\beta^1)SW^{1\top})^\top\leq 0.
\end{aligned}
\end{equation}
By choosing $\gamma=\frac{\underline{s}}{\mu}$, $S=P^{-1}$, $H=YS^{-1}$, and using \eqref{CO3p1}, we have
\begin{equation}\label{CO3p3}
\begin{aligned}
&P^{-1}\Big(L'+\frac{1}{2}(PW_{\min}+(\alpha^1+\beta^1)W^{1\top})\\&\times(PW_{\min}+(\alpha^1+\beta^1)W^{1\top})^\top\Big)P^{-1}\\&=\bar{R}+\frac{1}{2}\tilde{R}\tilde{R}^{\top}+\gamma S
\leq
\bar{R}+\frac{1}{2}\tilde{R}\tilde{R}^{\top}+\frac{1}{\mu}S^2\leq 0,
\end{aligned}
\end{equation}
where $L'$, $W_{min}$, $W^1$, $\alpha^1$, $\beta^1$ are as defined in Lemma \ref{corollary2}. Then, the inequality \eqref{CO3p3} is equivalent to the Schur complement hence to \eqref{CO22}. Since the conditions in Lemma \ref{corollary2} are satisfied, the conclusion follows immediately.  \end{proof}
\begin{remark}\label{rem4}
By choosing $\gamma=\frac{\underline{s}}{\mu}$, $S=P^{-1}$, $H=YS^{-1}$, we have $\underline{p}=\frac{1}{\overline{s}}$, $\overline{p}=\frac{1}{\overline{s}}$, and $\frac{1}{\gamma}\sqrt{\frac{\overline{p}}{\underline{p}}}=\frac{\mu}{\underline{s}}\sqrt{\frac{\overline{s}}{\underline{s}}}$. Then the optimization problem \eqref{OPr} is reformulated as solving the following optimization problem:
\begin{align}
\underset{S, Y, \mu}{\text{minimize}} 
\quad & \frac{\mu}{\underline{s}}\sqrt{\frac{\overline{s}}{\underline{s}}}\nonumber\\
\text{s.t.} \quad & \eqref{CO31}, \eqref{CO32}.\label{OPC}
\end{align}
According to Proposition \ref{proposition3} (see Remark \ref{rem3}), the state $x(t)$ converges to the neighborhood of the unknown equilibrium $x^*$, whose size is given by $\frac{4\epsilon\mu}{\underline{s}}\sqrt{\frac{\overline{s}}{\underline{s}}}$.
\end{remark}

We present the complete algorithm for designing the Hopfield NN-based controller for system \eqref{ND0}. The procedure starts by learning the unknown drift dynamics $f(x)$ in \eqref{ND0} using the neural network \eqref{NN}, which is trained by solving the optimization problem \eqref{aca}, yielding an upper bound $\epsilon$ for the approximation error $\varepsilon(x)$. Next, compute the parameters $\alpha^1$ and $\beta^1$ using the trained parameters $W^1$, $b^1$, and the input $x$ from the set $\mathcal{D}$. Then, determine the parameter $W_{u\min}^{2}$ of the NN controller by solving the problem $\underset{W_u^{2}}{\text{min}}\|W^2+gW_u^{2}\|$. Finally, solve the optimization problem \eqref{OPC} to find  $\mu$, $\overline{s}$, $\underline{s}$, and calculate the control gain $H$ by $H=YS^{-1}$. The parameter $\gamma$ is determined by $\gamma=\frac{\underline{p}}{\mu}$. The controller is given by $u(x)=Hx+W_{u\min}^{2}w^1$,  and the neighborhood is defined as $\|x-x^*\|\leq \frac{4\epsilon \mu}{\underline{s}}\sqrt{\frac{\overline{s}}{\underline{s}}}$.
The entire procedure is summarized in Algorithm \ref{alg:HNN_controller}.
\begin{algorithm}[H]
\caption{Hopfield Neural Network Based Controller Design}
\label{alg:HNN_controller}
\begin{algorithmic}[1]
\REQUIRE NODE system \eqref{NDL}, input set $\mathcal{D}$, trained parameters $W^{i+1}$, $b^{i+1}$
\ENSURE Control gain $H$, NN parameter $W_{u\min}^{2}$
\vspace{0.5em}
\STATE \textbf{Step 1: Learn System Dynamics}
\STATE Train the unknown drift dynamics $f(x)$ by solving \eqref{aca}.
\vspace{0.5em}
\STATE \textbf{Step 2: Compute Parameters $\alpha^1$, $\beta^1$}
\STATE Using the trained parameters $W^1$, $b^1$, and inputs $x \in \mathcal{D}$, compute $\alpha^1$, $\beta^1$.
\STATE \textbf{Step 3: Determine Neural Network Parameter $W_{u\min}^{2}$}
\STATE Solve $\min \|W^2 +gW_u^{2}\|$ to compute $W_{u\min}^{2}$.
\STATE \textbf{Step 4: Solve Optimization Problem for Control Parameters}
\STATE Solve the optimization problem \eqref{OPC} to determine $\mu$, $S$, and $Y$.

\STATE \textbf{Step 5: Compute Control Gain and the size of neighborhood}
\STATE Compute the control gain as $H = YS^{-1}$.
\STATE Construct the controller as $u(x)=Hx+W_{u\min}^{2}w^1$.
\STATE Compute $\gamma$ as $\gamma=\frac{\underline{s}}{\mu}$.
\STATE \textbf{Output:} Control gain $H$, NN parameter $W_{u\min}^{2}$, the size of neighborhood $\frac{4\epsilon \mu}{\underline{s}}\sqrt{\frac{\overline{s}}{\underline{s}}}$.
\end{algorithmic}
\end{algorithm}

\subsubsection{‌Multilayer neural network case}
By applying Theorem \ref{theorem1}, the following result establishes the contractivity of the previously non-contractive NODE system \eqref{NDL}.
\begin{theorem}\label{theorem2}
Suppose there exist a positive definite matrix $P$, positive constants $\underline{p}$, $\overline{p}$, $\gamma$, and a matrix $H$, such that \eqref{TH110} and the following condition 
\begin{equation}\label{TH22}
\begin{aligned}
&U_1^\top
\begin{bmatrix}
     (A+gH)^\top P+P(A+gH)+\gamma P &PW_{\min}\\
   \ast & 0
\end{bmatrix}U_1+\\&U_2^\top\begin{bmatrix}
 -2Q_1^\top Q_2&Q_1^\top+Q_2^\top
 \\
\ast&-2I
\end{bmatrix}U_2\leq 0,
\end{aligned}
\end{equation}
hold, where $U_1$, $U_2$, $W^{k+1}$, $Q_1$, $Q_2$ are given in Theorem \ref{theorem1}. Then, the learned affine system \eqref{NNC3} is contractive in $\mathcal{C}$ under the controller specified by \eqref{NNC11}.
\end{theorem}

Similarly to the Hopfield neural network case, i.e., Theorem \ref{the21}, the next theorem reformulates the contractive condition \eqref{TH22} into a convex LMI.

\begin{theorem}\label{theorem3}
Suppose there exist a fixed positive constant $\underline{s}$, positive definite matrices $S$, positive constants $\mu$,  $\overline{s}$, and a matrix $Y$, such that \eqref{CO31} and the following condition
\begin{equation}\label{TH32}
\begin{aligned}
&\bar{U}_1^\top
\bbm{
-\mu I & S & 0 \\
\ast &  (AS+gY)^\top + AS + gY & W_{\min} \\
\ast & \ast & 0
}
\bar{U}_1 + \\
&\bar{U}_2^\top
\bbm{
0 & 0 & 0 \\
\ast & \bar{X} &\tilde{X} \\
\ast & \ast & -2I
}
\bar{U}_2 \leq 0, \quad \forall k \geq 2, 
\end{aligned}
\end{equation}
hold, where $\bar{X}=\bbm{-2\underline{s} \underline{\lambda}S&0\\0&-2Q^{a\top}_1 Q^a_2}$, $\underline{\lambda}$ is the minimum eigenvalue of $Q^{1\top}_1 Q^1_2$,  $\tilde{X}=\bbm{S(Q^1_1+Q^1_2)^\top&0\\0&(Q^a_1+Q^a_2)^\top}$, $Q^a_1=\mathrm{BD}(K_1^2W^2, \dots, K_1^{k}W^{k})$, $Q^a_2=\mathrm{BD}(K_2^2W^2, \dots, K_2^{k}W^{k})$, $\bar{U}_1$, $\bar{U}_2$, $W^{k+1}$ are given in Theorem \ref{theorem11}. Then, the learned affine system \eqref{NNC3} is contractive in $\mathcal{C}$ under the controller specified by \eqref{NNC10} with $H=YS^{-1}$.
\end{theorem}
\begin{proof}
Observe that $\bar{U}_1=\bbm{I_{m_{k+1}\times m_{k+1}}&0\\0&U_1}$, $\bar{U}_2=\bbm{I_{m_{k+1}\times m_{k+1}}&0\\0&U_2}$. Following a similar approach as \eqref{TH121P1}-\eqref{TH121P4} in the proof of Theorem \ref{theorem11}, but substituting $A$ and $P$ with $AS + gY$ and $S$, respectively, we can rewrite \eqref{TH32} as
\begin{equation}\label{TH3P1}
\begin{aligned}
&U_1^\top
\bbm{
(AS+gY)^\top + AS + gY+\frac{1}{\mu}S^2 & W_{\min} \\
\ast & 0
}
U_1 + \\
&U_2^\top
\bbm{
\bar{X} &\tilde{X} \\
\ast & -2I
}
U_2 \leq 0, 
\end{aligned}
\end{equation}
We can express the left hand side of \eqref{TH3P1} in the following form:
\begin{equation}\label{TH3P40}
\begin{aligned}
&\Big(\bbm{
S&0\\0&I
}U_1\Big)^\top
\bbm{
X & S^{-1}W_{\min} \\
\ast & 0
}
\bbm{
S&0\\0&I
}U_1 + \\
&\Big(\bbm{
S&0&0&0\\\ast&I&0&0\\\ast&\ast&I&0\\\ast&\ast&\ast&I
}U_2\Big)^\top\\&\times
\bbm{\bbm{-2\underline{s}\underline{\lambda}S^{-1}&0\\0&-2Q^{a\top}_1 Q^a_2}
 &*\\
\bbm{(Q^1_1+Q^1_2)^\top&0\\0&(Q^a_1+Q^a_2)^\top}^\top  & -2I
}
\\&\times\bbm{
S&0&0&0\\\ast&I&0&0\\\ast&\ast&I&0\\\ast&\ast&\ast&I
}U_2, 
\end{aligned}
\end{equation}
where $X=(S^{-1}A + S^{-1}gYS^{-1})^\top + S^{-1}A + S^{-1}gYS^{-1}+\frac{1}{\mu}I$.

Note that
\begin{equation}\label{TH3P2}
\begin{aligned}
&U_1\bbm{
S&0&0\\\ast&I&0\\\ast&\ast&I
}=\bbm{
S&0\\0&I
}U_1,
\end{aligned}
\end{equation}
and
\begin{equation}\label{TH3P3}
\begin{aligned}
&U_2\bbm{
S&0&0\\\ast&I&0\\\ast&\ast&I
}=\bbm{
S&0&0&0\\\ast&I&0&0\\\ast&\ast&I&0\\\ast&\ast&\ast&I
}U_2,
\end{aligned}
\end{equation}
Using \eqref{TH3P2}, \eqref{TH3P3}, we can rewrite \eqref{TH3P4} as:
\begin{equation}\label{TH3P4}
\begin{aligned}
&\bbm{
S&0&0\\\ast&I&0\\\ast&\ast&I
}\Big(U_1^\top
\bbm{
X & S^{-1}W_{\min} \\
W_{\min}^{\top}S^{-1} & 0
}
U_1 +\\
& U_2^\top
\bbm{\bbm{-2\underline{s}\underline{\lambda}S^{-1}&0\\0&-2Q^{a\top}_1 Q^a_2}
 &\ast\\
\bbm{(Q^1_1+Q^1_2)^\top&0\\0&(Q^a_1+Q^a_2)^\top}^\top  & -2I
}
U_2\Big)\\&\times\bbm{
S&0&0\\\ast&I&0\\\ast&\ast&I
}.
\end{aligned}
\end{equation}
Using $P=S^{-1}$, $H=YS^{-1}$, we can rewrite \eqref{TH3P4} as
\begin{equation}\label{TH3P5}
\begin{aligned}
&\bbm{
P^{-1}&0&0\\\ast&I&0\\\ast&\ast&I
}\times\\&\Big(U_1^\top
\bbm{
(A+gH)^\top P+P(A+gH)+\frac{1}{\mu}I & PW_{\min} \\
\ast & 0
}
U_1 \\
&+ U_2^\top
\bbm{\bbm{-2\underline{s}\underline{\lambda}P&0\\0&-2Q^{a\top}_1 Q^a_2}
 &\ast\\
\bbm{(Q^1_1+Q^1_2)^\top&0\\0&(Q^a_1+Q^a_2)^\top}^\top  & -2I
}
U_2\Big)\\&\times\bbm{
P^{-1}&0&0\\\ast&I&0\\\ast&\ast&I
}.
\end{aligned}
\end{equation}
Since $\bbm{
P^{-1}&0&0\\\ast&I&0\\\ast&\ast&I
}$ is invertible, using \eqref{TH3P5}, the inequality \eqref{TH3P1} is equivalent to
\begin{equation}\label{TH3P6}
\begin{aligned}
&U_1^\top
\bbm{
(A+gH)^\top P+P(A+gH)+\frac{1}{\mu}I & PW_{\min} \\
\ast & 0
}
U_1 +\\
& U_2^\top
\bbm{\bbm{-2\underline{s}\underline{\lambda}P&0\\0&-2Q^{a\top}_1 Q^a_2}
 &\ast\\
\bbm{(Q^1_1+Q^1_2)^\top&0\\0&(Q^a_1+Q^a_2)^\top}^\top  & -2I
}
U_2\leq 0.
\end{aligned}
\end{equation}
Using $\underline{s}P\geq I$, $\gamma=\frac{\underline{s}}{\mu}$, we have 
\begin{equation}\label{TH3P7}
\begin{aligned}
&U_1^\top\bbm{\gamma P&0\\0&0}U_1=\bbm{\gamma P&0&0\\\ast&0&0\\\ast&\ast&0}\leq \bbm{\frac{1}{\mu}I&0&0\\\ast&0&0\\\ast&\ast&0}\\&=U_1^\top\bbm{\frac{1}{\mu}I&0\\0&0}U_1,
\end{aligned}
\end{equation}
and
\begin{equation}\label{TH3P8}
\begin{aligned}
&U_2^\top\bbm{-2Q_1^{1\top}Q_2^1&0&0&0\\\ast&0&0&0\\\ast&\ast&0&0\\\ast&\ast&\ast&0}U_2=\bbm{-2Q_1^{1\top}Q_2^1&0&0&0\\\ast&0&0&0\\\ast&\ast&0&0\\\ast&\ast&\ast&0}\\&\leq \bbm{-2\underline{s}\underline{\lambda}P&0&0&0\\\ast&0&0&0\\\ast&\ast&0&0\\\ast&\ast&\ast&0}=U_2^\top\bbm{-2\underline{s}\underline{\lambda}P&0&0&0\\\ast&0&0&0\\\ast&\ast&0&0\\\ast&\ast&\ast&0}U_2.
\end{aligned}
\end{equation}
Then inequity \eqref{TH22} in Theorem \ref{theorem1} is satisfied by applying inequities \eqref{TH3P6}, \eqref{TH3P7}, and \eqref{TH3P8}. The conditions in Theorem \ref{theorem2} are met and the
conclusion of the theorem follows directly. 
\end{proof}
\begin{remark}\label{rem5}
The controller design problem can be formulated as the following optimization problem (cf. \eqref{OPC}):
\begin{align}
\underset{S, Y, \mu}{\text{minimize}} 
\quad & \frac{\mu}{\underline{s}}\sqrt{\frac{\overline{s}}{\underline{s}}}\nonumber\\
\text{s.t.} \quad & \eqref{CO31}, \eqref{TH32}.\label{OPC1}
\end{align}
\end{remark}

We provide the complete algorithm for designing the Multilayer NN-based controller for system \eqref{ND0} in Algorithm \ref{alg:MNN_controller}. The key distinction from Algorithm \ref{alg:HNN_controller} lies in the computation of the parameters $\alpha^{i+1}$ and $\beta^{i+1}$, which are derived from the trained parameters $W^{i+1}$, $b^{i+1}$, $\phi^{i+1}$, and $w^i$. Additionally, the optimization problem \eqref{OPC1} is solved to determine $\mu$ and $\overline{s}$.
\begin{algorithm}[H]
\caption{Multilayer Neural Network Based Controller Design}
\label{alg:MNN_controller}
\begin{algorithmic}[1]
\REQUIRE NODE system \eqref{NDL}, input set $\mathcal{D}$, trained parameters $W^{i+1}$, $b^{i+1}$
\ENSURE Control gain $H$, NN parameter $W_{u\min}^{k+1}$
\vspace{0.5em}
\STATE \textbf{Step 1: Learn System Dynamics}
\STATE Train the unknown drift dynamics $f(x)$ by solving \eqref{aca}.
\STATE \textbf{Step 2: Compute Parameters $\alpha^{i+1}$, $\beta^{i+1}$}
\STATE Using the trained parameters $W^{i+1}$, $b^{i+1}$, $\phi^{i+1}$, and $w^{i}$ compute $\alpha^{i+1}$, $\beta^{i+1}$.
\STATE \textbf{Step 3: Determine Neural Network Parameter $W_{u\min}^{k+1}$}
\STATE Solve $\min \|W^{k+1} +gW_u^{k+1}\|$ to compute $W_{u\min}^{k+1}$.
\STATE \textbf{Step 4: Solve Optimization Problem for Control Parameters}
\STATE Solve the optimization problem \eqref{OPC1} to determine $\mu$, $S$, and $Y$.
\STATE \textbf{Step 5: Compute Control Gain and the size of neighborhood}
\STATE Compute the control gain as $H = YS^{-1}$.
\STATE Construct the controller as $u(x)=Hx+W_{u\min}^{k+1}w^k$.
\STATE Compute $\gamma$ as $\gamma=\frac{\underline{s}}{\mu}$.
\STATE \textbf{Output:} Control gain $H$, NN parameter $W_{u\min}^{k+1}$, the size of neighborhood $\frac{4\epsilon \mu}{\underline{s}}\sqrt{\frac{\overline{s}}{\underline{s}}}$.
\end{algorithmic}
\end{algorithm}
\section{Simulation setup}\label{secss}
In this section, we provide two numerical examples to demonstrate the proposed method. In the first example, we design an NN-based controller $u(x)$, where the NN component $u_{nn}(x)$ effectively cancels out most of the system's nonlinearity. The second example showcases a scenario where the nonlinearity is difficult to cancel using $u_{nn}(x)$. Nevertheless, due to the incremental sector bound property of $Z_{nn}(x)$, the contractivity of the NODE system can still be achieved, and the trajectories of the original system convergence to a small neighborhood of its equilibrium. In both examples, we use the tanh function as the activation function.

\subsubsection{‌Hopfield (Single-layer) neural network case} 
\begin{example}\label{example2}
Consider the following pendulum dynamics
\begin{equation}\label{E2}
\left\{\begin{matrix}
\dot{x}_1=x_2,\\
\dot{x}_2=-\frac{g}{l}\sin(x_1)-\frac{k}{m} x_2+\frac{1}{ml^2}u(x),
\end{matrix}\right.
\end{equation} 
where $m$ is the mass to be balanced, $l$ is the distance from the base to the center of mass of the balanced body, $k$ is the coefficient of rotational friction, and $g$ is the acceleration due to gravity. The states $x_1$, $x_2$ are the angular position and velocity, respectively, $u(x)$ is the applied torque. Suppose that the parameters are $g=9.81$, $m$ = 0.15, $l=5$. To better illustrate our approach, we set $k=0$, which implies that the uncontrolled system converges to a limit cycle rather than an equilibrium. 

The data points $x(t_k)$ and $\dot{x}(t_k)$ are obtained from the pendulum dynamics \eqref{E2} where $x(t_k)$ lies within the range $\bbm{-2,&2}^2$. We learn the drift dynamics of the system, i.e., $f(x)=\bbm{x_2\\-1.962\sin(x_1)}$ by solving the problem \eqref{aca} by an NN with a single-hidden-layer containing 5 neurons. We have $A=\bbm{   -0.0000 &   0.9996\\
   -1.7654  &  0.0012}$, $W^1=  \bbm{
-0.0365 &  -0.0220 &   -0.0061  &  -0.0180 &   0.0169\\
    0.0100& 0.0123 &  0.0010 & -0.0063 &   0.0167
}^\top$, $W^2=\bbm{
          0.0000 &   0.0120 &  -0.0029 &  -0.0054 &   0.0101\\
    0.0351 &   0.0267  &  0.0051  &  0.0134  & -0.0181 
}$, $b^1=\bbm{
       -0.4525&
   -0.0774&
   -0.1071&
   -0.3386&
    0.1009
}^\top\times10^{-3}$, $b^2= \bbm{
           -0.0000&
    0.0016
}^{\top}$. After 5000 epochs, the upper bound $\epsilon$ of the approximation error $\varepsilon(x)$ is 0.002302. 
Since for all $x_1,x_2\in\bbm{-2,&2}$, we have
\begin{equation}\label{E22}
0.9978\leq\frac{\varphi^1 (x)-\varphi^1 (y)}{x-y}\leq 1.
\end{equation} 
By solving $W_{\min}^{1'}=\underset{W^{1'}}{\text{argmin}} \|W^1+gW^{1'}\|$, we obtain $W_{\min}^{1'}=\bbm{
     0 &  0 &   0 &  0  &  0\\
    -0.0351 &   -0.0267  &  -0.0051  &  -0.0134  & 0.0181          
}$, and $W_{\min}=\bbm{
      -0.0000  &  0.0120 &  -0.0029  & -0.0054  &  0.0101\\  
   0 & 0 &  0  &  0  & 0       
}$.
Applying the incremental sector bound \eqref{E22}, the $\bar{R}$ and $R'$ in the contractive condition \eqref{CO32} are given by $\bar{R}=-102850W_{\min} W_{\min}^\top+A^\top S+SA+gY+Y^\top g^\top $, $\tilde{R}=453.5455 W_{\min} +0.0044SW_1^{\top}$. 

By solving the optimization problem \eqref{OPC}, we have $S=\bbm{0.4807  & -0.2379
   \\-0.2379  &  1.0928}$, $Y=\bbm{1.3799  & -5.7854}$, $H=\bbm{0.2808 &  -5.2330}$, $\mu=0.61168$, where the state ultimately converges to $\frac{4\epsilon \mu}{\underline{s}}\sqrt{\frac{\overline{s}}{\underline{s}}}=0.0239$. Figure \ref{fig:21} illustrates the trajectories of the system with four distinct initial conditions $\bbm{1&1}^{\top}$, $\bbm{-1&-1}^{\top}$, $\bbm{-1.5&1.5}^{\top}$, $\bbm{1.5&-1.5}^{\top}$. 

   \begin{figure}
    \centering
    \includegraphics[width=1.1\linewidth]{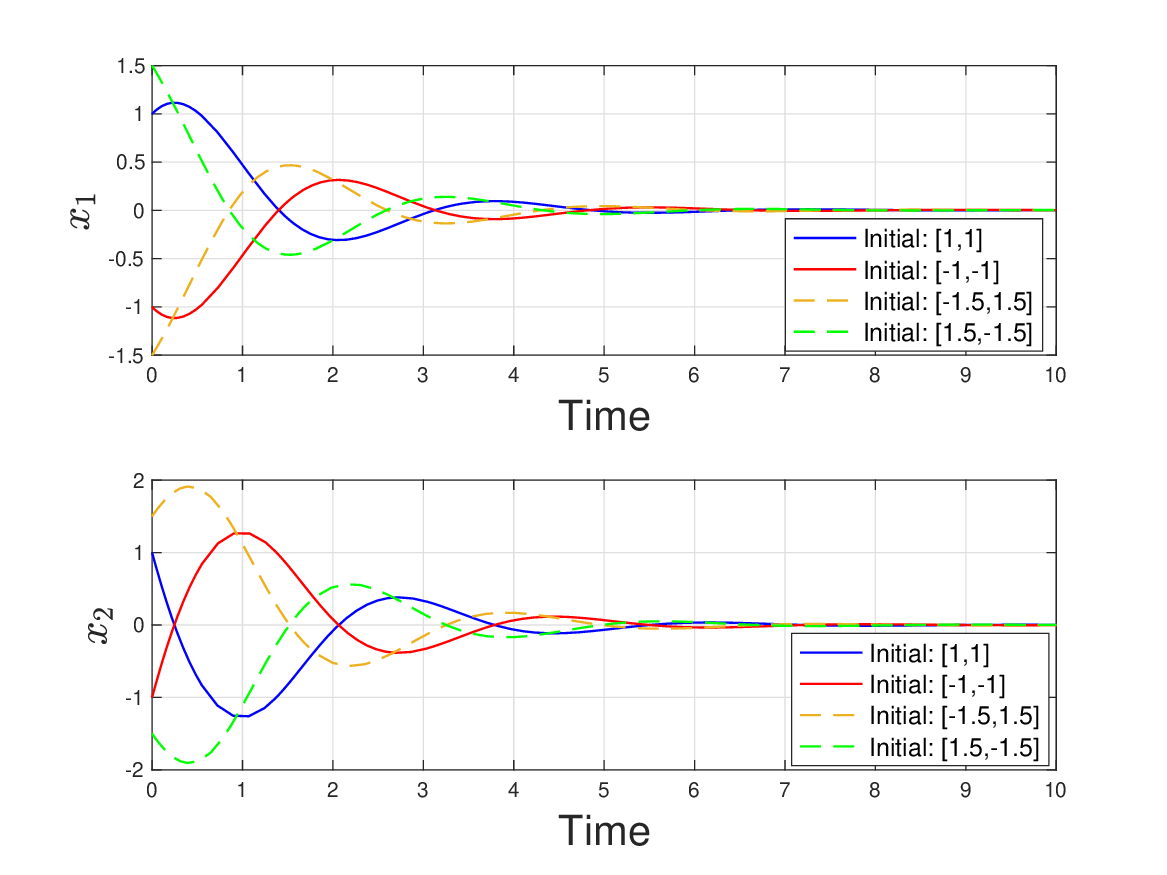}
    \caption{The plot shows the trajectories of the controlled system, $x_i$ in Example \ref{example2}
initialized at $\bbm{1&1}^{\top}$, $\bbm{-1&-1}^{\top}$, $\bbm{-1.5&1.5}^{\top}$, $\bbm{1.5&-1.5}^{\top}$.}
    \label{fig:21}
\end{figure}
\end{example}

\subsubsection{‌Multilayer neural network case}
\begin{example}\cite[Ex. 2]{samanipour2023stability}\label{example3}
Consider a wheeled
vehicle path following system 
\begin{equation}\label{E3}
\left\{\begin{matrix}
\dot{d}_e=\nu\sin(\theta_e)\\
\dot{\theta}_e=\omega-\frac{\nu\kappa(s)\cos(\theta_e)}{1-\kappa(s)d_e}.
\end{matrix}\right.
\end{equation}   
The variables are defined as follows: 
$\theta_e$ represents the heading error, $d_e$ denotes the lateral error, $\nu$ is the linear velocity, 
$\kappa(s)$ describes the curvature of the reference path, and $\omega$ represents the angular velocity, which we consider as the control input. We assumed that the target path is a unit circle, $\kappa(s)=0.5$, and the linear velocity $\nu=1$.
In this example, since the controller is only present in the $\dot{\theta}_e$ channel, it is difficult to eliminate the nonlinearity in the $\dot{d}_e$ channel, i.e., the term $\nu\sin(\theta_e)$ cannot be canceled by the controller. 

In \cite{samanipour2023stability}, a single-hidden-layer ReLU NN with 50 neurons is used to model the system without considering the approximation error, leading to a high-dimensional LMI. In contrast, we employ a two-hidden-layer NN with 5 neurons for each layer to represent the drift dynamics. Despite this difference in architecture, our NN achieves the same level of approximation accuracy as the 50-neuron single-hidden-layer ReLU NN but with a lower dimensional LMI. For a two-hidden layer NN, by applying Theorem \ref{theorem3}, inequality \eqref{TH32} can be reduced to the following inequality.
\begin{equation}\label{le31}
\begin{aligned}
&\begin{bmatrix}
-\mu I&S &0&0\\\ast&\bar{R}&S(Q^1_1+Q^1_2)^\top&W_{\min}\\\ast&\ast&-2Q_1^{2\top}Q_2^2-2I&(Q_1^2+Q_2^2)^\top\\
\ast&\ast&\ast& -2I
\end{bmatrix}\leq 0,
\end{aligned}
\end{equation}
where $\bar{R}=(AS+gY)^\top + AS + gY-2\underline{s} \underline{\lambda}S$, $\underline{\lambda}$ is the minimum eigenvalue of $Q^{1\top}_1 Q^1_2$.

The data points $x(t_k)$ and $\dot{x}(t_k)$ are obtained from the wheeled vehicle path
following system \eqref{E3} where $x(t_k)$ lies within the range $\bbm{-1,&1}^2$. We learn the drift dynamics of the system, i.e., $f(x)=\bbm{\sin(\theta_e)\\-\frac{0.5\cos(\theta_e)}{1-0.5d_e}}$ by solving the problem \eqref{aca}. We have $A=\bbm{   -0.0011  &  0.9023\\
   -0.2454  & -0.0031}$, $W^1=  \bbm{0.0002 &   -0.0100 &   -0.0005 &   0.0040 &  -0.0114\\
   0.0122  & 0.0122  &  -0.0009  &  0.0062  & -0.0124   
}^\top$, $W^2=\bbm{    0.0057  &  0.0002 &  -0.0091   & 0.0008  & -0.0066\\
   -0.0011 &   0.0046 &  -0.0092  &  0.0105  & -0.0109\\
    0.0163  &  0.0147 &   0.0107 &  -0.0015  & -0.0066\\
   -0.0096  &  0.0111  & -0.0014  &  0.0031   &-0.0076\\
    0.0201&   -0.0088  & -0.0045  & -0.0231 &  -0.0043
}$, $W^3=\bbm{   -0.0048   & 0.0146  &  0.0179  &  0.0011  &  0.0025\\
    0.0072 &  -0.0170  &  0.0142  &  0.0174 &  -0.0113
}$, $b^1=\bbm{0.0447&
   -0.1820&
   -0.1153&
   -0.0527&
    0.0181
}^\top\times10^{-3}$, $b^2=\bbm{
         -0.0031&
    0.0073&
   -0.0061&
   -0.0075&
    0.0049
}^\top$, $b^3= \bbm{
            0.0013&
   -0.4621
}^{\top}$. After 5000 epochs, the upper bound $\epsilon$ of the approximation error $\varepsilon(x)$ is 0.008532. 
Since for all $x,y\in\bbm{-1,&1}$, we have
\begin{equation}\label{E32}
\begin{matrix}
 0.9994\leq\frac{\varphi^1 (x)-\varphi^1 (y)}{x-y}\leq 1,\\
0.9999\leq\frac{\varphi^2 (x)-\varphi^2 (y)}{x-y}\leq 1.
\end{matrix}
\end{equation} 
By solving $W_{\min}^{3'}=\underset{W^{3'}}{\text{argmin}} \|W^3+gW^{3'}\|$, we obtain $W_{\min}^{3'}=\bbm{
   0 &   0 &  0 &   0 & 0\\
   -0.0072 &  0.0170  &  -0.0142  &  -0.0174  & 0.0113
}$, and $W_{\min}=\bbm{
   -0.0048 &   0.0146   & 0.0179   & 0.0011 &   0.0025\\
   0 & 0 &  0  &  0  & 0       
}$. In this case, provided that $\underline{\lambda}$ is sufficiently small, the term $\bar{R}$ in \eqref{le31} can be approximated as $\bar{R}\approx (AS+gY)^\top + AS + gY$.

By solving the optimization problem 
\begin{align}
\underset{S, Y, \mu}{\text{minimize}} 
\quad & \frac{\mu}{\underline{s}}\sqrt{\frac{\overline{s}}{\underline{s}}}\nonumber\\
\text{s.t.} \quad & \eqref{CO31}, \eqref{le31},\label{OPC2}
\end{align} we have $S=\bbm{1.9914 &  -0.7703\\
   -0.7703  &  12.1027}$, $Y=\bbm{
-2.6214 &  -18.6585
}$, $H=\bbm{
-1.9610  & -1.6665
}$, $\mu=8.0543 \times 10^{-15}$, where the state ultimately converges to $\frac{4\epsilon \mu}{\underline{s}}\sqrt{\frac{\overline{s}}{\underline{s}}}=3.5567 \times 10^{-16}$. Figure \ref{fig:31} illustrates the trajectories of the system with four distinct initial conditions $\bbm{0.8&0.8}^{\top}$, $\bbm{-0.8&-0.8}^{\top}$, $\bbm{-0.5&0.5}^{\top}$, $\bbm{0.5&-0.5}^{\top}$. 

\begin{figure}
    \centering
    \includegraphics[width=1.1\linewidth]{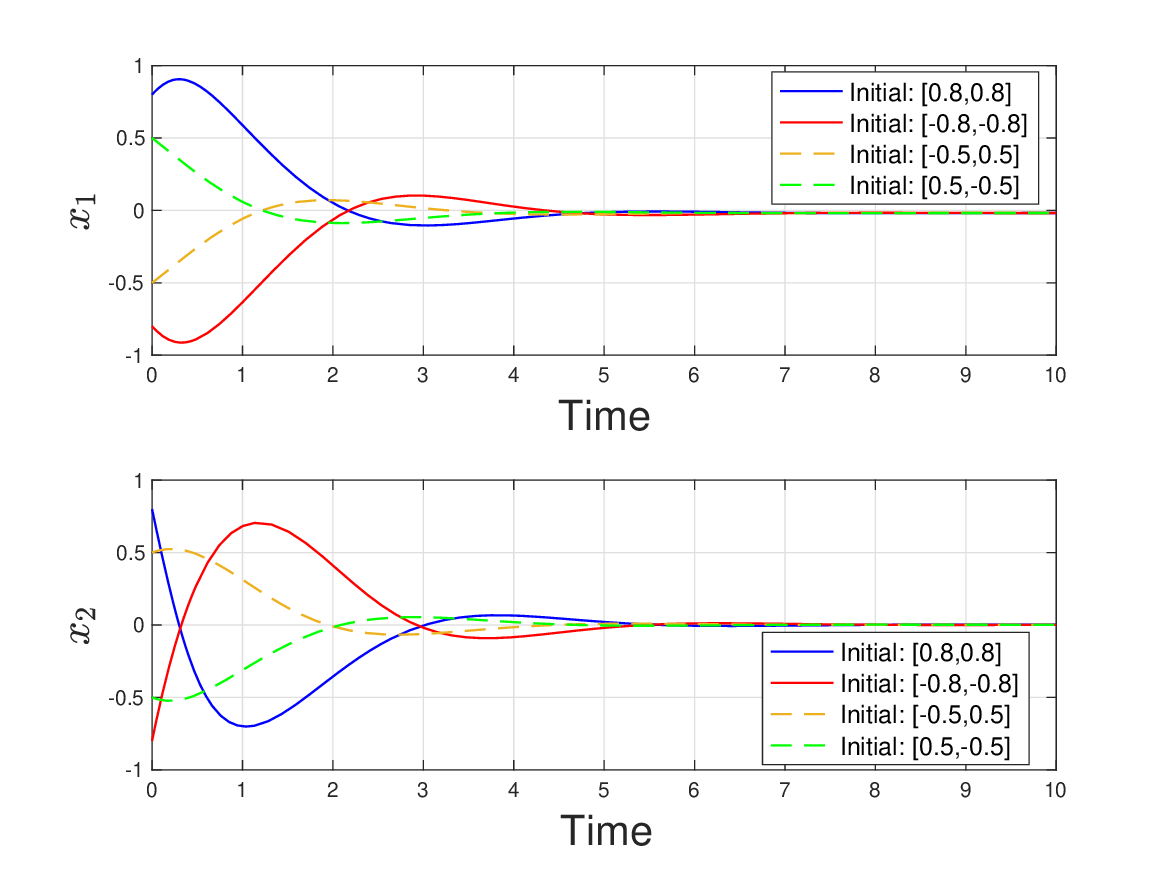}
    \caption{The plot shows the trajectories of the controlled system, $x_i$ in Example \ref{example3}
initialized at $\bbm{0.8&0.8}^{\top}$, $\bbm{-0.8&-0.8}^{\top}$, $\bbm{-0.5&0.5}^{\top}$, $\bbm{0.5&-0.5}^{\top}$.}
    \label{fig:31}
\end{figure}
\end{example}

\section{CONCLUSION}\label{secc}
This paper explores the contraction properties and controller design for unknown nonlinear affine systems through NN approximation. The autonomous system is approximated by a NODE system with incrementally sector-bounded activation functions. To ensure the contractivity of the NODE system, we validate the NN’s weights using a contraction condition. The contractivity of the NODE system indicates that the trajectories of the original autonomous system converge to a neighborhood of the unknown equilibrium point, with the neighborhood size determined by the approximation error. The NODE system can be transformed into a continuous-time Hopfield NN by utilizing a single-layer NN. If the contraction conditions are not satisfied, we propose a contractive NN controller based on the approximated NODE system, comprising an NN component to compensate system's nonlinearities and a linear component whose control gain ensures the NODE system’s contractivity. We prove that the contractive NN controller guarantees the trajectories of the nonlinear affine system converge to a specified neighborhood of the equilibrium point. Additionally, an optimization framework is introduced to minimize the size of this neighborhood. Future work will focus on replacing the traditional LMI approach with Physics-Informed Neural Networks (PINNs) for constructing Lyapunov functions beyond the quadratic form. This is expected to enhance the flexibility of contraction analysis and extend the applicability of the framework to a broader range of systems.

\end{document}